\documentclass[orivec]{llncs}

\usepackage{a4,a4wide}
\usepackage{algorithm}
\usepackage{algorithmic}
\usepackage{longtable}
\usepackage{amsmath}
\usepackage{amssymb}
\usepackage{graphicx}
\usepackage{float}
\usepackage{url}

\newcommand{\F}{\mathcal{F}}
\newcommand{\R}{\mathcal{R}}
\newcommand{\Pe}{\mathcal{P}}
\newcommand{\Pa}{\mathcal{P}}
\newcommand{\D}{\mathcal{D}}
\newcommand{\He}{\mathcal{H}}
\newcommand{\I}{\mathcal{I}}

\newcommand{\St}{\mathcal{S}}

\newcommand{\M}{\mathcal{M}}
\newcommand{\Mo}{\mathcal{M}}

\newcommand{\muplus}{\makebox[1em][c]{$\cup$} \hspace*{-1em} \makebox[1em][c]{$\cdot$}}

\newtheorem{theo}{Theorem}

\begin{document}
\title{A Uniform Fixpoint Approach to the Implementation of Inference Methods for Deductive Databases}
\author{Andreas Behrend}
\institute{University of Bonn,\\Institute of Computer Science\,III,R{\"o}merstra{\ss}e 164,\\
 53117 Bonn, Germany\\
\email{behrend@cs.uni-bonn.de}}

\maketitle

%%%%%%%%%%%%%%%%%%%%%%%%%%%%%%%%%%%%%%%%%%%%%%%%%%%%%%%%%%%%%%%%%%%%%%%%%%%%
%
%       Abstract
%
%%%%%%%%%%%%%%%%%%%%%%%%%%%%%%%%%%%%%%%%%%%%%%%%%%%%%%%%%%%%%%%%%%%%%%%%%%%%%
\begin{abstract}
Within the research area of deductive databases three different
database tasks have been deeply investigated: query evaluation,
update propagation and view updating. Over the last thirty years
va\-ri\-ous inference mechanisms have been proposed for realizing
these main functionalities of a rule-based system. However, these
inference mechanisms have been rarely used in commercial DB
systems until now. One important reason for this is the lack of a
uniform approach well-suited for implementation in an SQL-based
system. In this paper, we present such a uniform approach in form
of a new version of the soft consequence operator. Additionally,
we present improved transformation-based approaches to query
optimization and update propagation and view updating which are
all using this operator as underlying evaluation mechanism.
\end{abstract}

%%%%%%%%%%%%%%%%%%%%%%%%%%%%%%%%%%%%%%%%%%%%%%%%%%%%%%%%%%%%%%%%%%%%%%%%%%%%%
%
%       Introduction
%
%%%%%%%%%%%%%%%%%%%%%%%%%%%%%%%%%%%%%%%%%%%%%%%%%%%%%%%%%%%%%%%%%%%%%%%%%%%%%

\section{Introduction}

The notion deductive database refers to systems capable of
inferring new knowledge using rules. Within this research area,
three main database tasks have been intensively studied:
(recursive) query evaluation, update propagation and view
updating. Despite of many proposals for efficiently performing
these tasks, however, the corresponding methods have been
implemented in commercial products (such as, e.g., Oracle or DB2)
in a very limited way, so far. One important reason is that many
proposals employ inference methods which are not directly suited
for being transferred into the SQL world. For example, proof-based
methods or instance-oriented model generation techniques (e.g.
based on SLDNF) have been proposed as inference methods for view
updating which are hardly compatible with the set-oriented
bottom-up evaluation strategy of SQL.

In this paper, we present transformation-based methods to query
optimization, update propagation and view updating which are
well-suited for being transferred to SQL. Transformation-based
approaches like Magic Sets~\cite{br86} automatically transform a
given database schema into a new one such that the evaluation of
rules over the rewritten schema performs a certain database task
more efficiently than with respect to the original schema. These
approaches are well-suited for extending database systems, as new
algorithmic ideas are solely incorporated into the transformation
process, leaving the actual database engine with its own
optimization techniques unchanged. In fact, rewriting techniques
allow for implementing various database functionalities on the
basis of one common inference engine. However, the application of
transformation-based approaches with respect to stratifiable
views~\cite{prz88} may lead to unstratifiable recursion within the
rewritten schemata. Consequently, an elaborate and very expensive
inference mechanism is generally required for their evaluation
such as the alternating fixpoint computation or the residual
program approach proposed by van Gelder~\cite{vg93} resp.
Bry~\cite{Bry89}. This is also the case for the kind of recursive
views proposed by the SQL:1999 standard, as they cover the class
of stratifiable views.

%    \begin{figure}[t]
%     \begin{center}
%       \includegraphics[width=9cm]{inferencesUP.eps}
%       \caption{\label{variousinferences}Inference methods for update propagation}
%     \end{center}
%     \vspace*{-0.5cm}
%    \end{figure}

As an alternative, the soft consequence operator together with the
soft stratification concept has been proposed by the author
in~\cite{beh03} which allows for the efficient evaluation of Magic
Sets transformed rules. This efficient inference method is
applicable to query-driven as well as update-driven derivations.
Query-driven inference is typically a top-down process whereas
update-driven approaches are usually designed bottom-up. During
the last 6 years, the idea of combining the advantages of top-down
and bottom-up oriented inference has been consequently employed to
enhance existing methods to query optimization~\cite{beh05} as
well as update propagation~\cite{BM04} and to develop a new
approach to view updating. In order to handle alternative
derivations that may occur in view updating methods, an extended
version of the original soft consequence operator has to be
developed. In this paper, this new version is presented, which is
well-suited for efficiently determining the semantics of definite
and indefinite databases but remains compatible with the
set-oriented, bottom-up evaluation of SQL.

%%%%%%%%%%%%%%%%%%%%%%%%%%%%%%%%%%%%%%%%%%%%%%%%%%%%%%%%%%%%%%%%%%%%%%%%%%%%%
%
%       Base concepts
%
%%%%%%%%%%%%%%%%%%%%%%%%%%%%%%%%%%%%%%%%%%%%%%%%%%%%%%%%%%%%%%%%%%%%%%%%%%%%%

\section{Basic concepts}
\label{Basic concepts}
A Datalog \emph{rule} is a function-free clause of the form $H_1
\leftarrow L_{1} \wedge \dots \wedge L_{m}$ with  $m\geq 1$ where
$H_1$ is an atom denoting the rule's head, and $L_{1},\dots,
L_{m}$ are literals, i.e. positive or negative atoms, representing
its body. We assume all deductive rules to be \emph{safe}, i.e.,
all variables occurring in the head or in any negated literal of a
rule must be also present in a positive literal in its body. If $A
\equiv p(t_{1},\dots,t_{n})$ with $n\geq 0$ is a literal, we use
${\tt  vars}(A)$ to denote the set of variables occurring in A and
${\tt pred}(A)$ to refer to the predicate symbol p of A. If A is
the head of a given rule $R$, we use ${\tt pred}(R)$ to refer to
the predicate symbol of A. For a set of rules $\R$, {\tt
pred}($\R$) is defined as $\cup_{r \in \R}\{ {\tt pred}(r) \}$. A
\emph{fact} is a ground atom in which every $t_i$ is a constant.

A \emph{deductive database} $\D$ is a triple \mbox{$\langle
\F,\R,\I \rangle$} where $\F$ is a finite set of facts (called
\emph{base facts}), $\I$ is a finite set of integrity constraints
(i.e.,positive ground atoms) and $\R$ a finite set of rules such
that ${\tt pred}(\F)\cap {\tt pred}(\R) = \O$ and ${\tt
pred}(\I)\subseteq {\tt pred}(\F \cup \R)$. Within a deductive
database $\D$, a predicate symbol $p$ is called derived (view
predicate), if $p \in {\tt pred}(\R)$. The predicate $p$ is called
extensional (or base predicate), if $p \in {\tt pred}(\F)$. Let
$\He_\D$ be the Herbrand base of $\D=\langle \F,\R,\I \rangle$.
The set of all derivable literals from $\D$ is defined as the
well-founded model~\cite{VGRS91} for $(\F \cup \R)$: $\Mo_\D :=
I^+ \cup \neg \cdot I^-$ where $I^+,I^-\subseteq \He_\D$ are sets
of ground atoms and $\neg \cdot I^-$ includes all negations of
atoms in $I^-$. The set $I^+$ represents the positive portion of
the well-founded model while $\neg \cdot I^-$ comprises all
negative conclusions. The semantics of a database $\D=\langle
\F,\R,\I \rangle$ is defined as the well-founded model $\Mo_\D :=
I^+ \cup \neg \cdot I^-$ for $\F \cup \R$ if all integrity
constraints are satisfied in $\Mo_\D$, i.e., $\I \subseteq I^+$.
Otherwise, the semantics of $\D$ is undefined. For the sake of
simplicity of exposition, and without loss of generality, we
assume that a predicate is either base or derived, but not both,
which can be easily achieved by rewriting a given database.

Disjunctive Datalog extends Datalog by disjunctions of lite\-rals
in facts as well as rule heads. A disjunctive Datalog rule is a
function-free clause of the form $A_1 \vee \ldots \vee A_m
\leftarrow B_{1} \wedge \dots \wedge B_{n}$ with $m,n\geq 1$ where
the rule's head $A_1 \vee \ldots \vee A_m$ is a disjunction of
positive atoms, and the rule's body $B_{1},\dots, B_{n}$ consists
of literals, i.e. positive or negative atoms. A disjunctive fact
$f \equiv f_1 \vee \ldots \vee f_k$ is a disjunction of ground
atoms $f_i$ with $i\geq1$. $f$ is called definite if $i=1$. We
solely consider stratifiable disjunctive rules only, that is,
recursion through negative predicate occurrences is not
permitted~\cite{prz88}. A stratification partitions a given rule
set such that all positive derivations of relations can be
determined before a negative literal with respect to one of those
relations is evaluated. The semantics of a stratifiable
disjunctive databases $\D$ is defined as the perfect model state
$\Pe\Mo_\D$ of $\D$ iff $\D$ is consistent~\cite{beh07,FM92}.

%%%%%%%%%%%%%%%%%%%%%%%%%%%%%%%%%%%%%%%%%%%%%%%%%%%%%%%%%%%%%%%%%%%%%%%%%%%%%
%
%                  Transformation-Based Approaches
%
%%%%%%%%%%%%%%%%%%%%%%%%%%%%%%%%%%%%%%%%%%%%%%%%%%%%%%%%%%%%%%%%%%%%%%%%%%%%%
%
\section[Transformation-Based Approaches]{Transformation-Based Approaches}
\label{Transformation-Based Approaches}
The need for a uniform inference mechanism in deductive databases
is motivated by the fact that transformation-based approaches to
query optimization, update propagation and view updating are still
based on very different model generators. In this section, we
briefly recall the state-of-the-art with respect to these
transformation-based techniques by means of Magic Sets, Magic
Updates and Magic View Updates. The last two approaches have been
already proposed by the author in~\cite{BM04} and~\cite{BM08}.
Note that we solely consider stratifiable rules for the given
(external) schema. The transformed internal schema, however, may
not always be stratifiable such that more general inference
engines are required.

\subsection[Query Optimization]{Query Optimization}
\label{Query Optimization}
Various methods for efficient bottom-up evaluation of queries
against the intensional part of a database have been proposed,
e.g. Magic Sets~\cite{br86}, Counting~\cite{br91}, Alexander
method~\cite{rlk86}). All these approaches are rewriting
techniques for deductive rules with respect to a given query such
that bottom-up materialization is performed in a goal-directed
manner cutting down the number of irrelevant facts generated. In
the following we will focus on Magic Sets as this approach has
been accepted as a kind of standard in the field.

Magic Sets rewriting is a two-step transformation in which the
first phase consists of constructing an adorned rule set, while
the second phase consists of the actual Magic Sets rewriting.
Within an adorned rule set, the predicate symbol of a literal is
associated with an adornment which is a string consisting of
letters {\tt b} and {\tt f}. While {\tt b} represents a bound
argument at the time when the literal is to be evaluated, {\tt f}
denotes a free argument. The adorned version of the deductive
rules is constructed with respect to an adorned query and a
selected sip strategy~\cite{ram91} which basically determines for
each rule the order in which the body literals are to be evaluated
and which bindings are passed on to the next literal. During the
second phase of Magic Sets the adorned rules are rewritten such
that bottom-up materialization of the resulting database simulates
a top-down evaluation of the original query on the original
database. For this purpose, each adorned rule is extended with a
magic literal restricting the evaluation of the rule to the given
binding in the adornment of the rule's head. The magic predicates
themselves are defined by rules which define the set of relevant
selection constants. The initial values corresponding to the query
are given by the so-called magic seed. As an example, consider the
following stratifiable rules $\R$

 \begin{tabbing}
  \hspace*{2mm} ${\tt o(X,Y) \leftarrow} \neg {\tt p(Y,X) \wedge p(X,Y)} $\\
  \hspace*{2mm} ${\tt p(X,Y) \leftarrow e(X,Y)}$\\
  \hspace*{2mm} ${\tt p(X,Y) \leftarrow e(X,Z) \wedge p(Z,Y)}$
 \end{tabbing}
and the query {\tt ?-o(1,2)} asking whether a path from node 1 to
2 exists but not vice versa. Assuming a full left-to-right sip
strategy, Magic Sets yields the following deductive rules
$\R_{ms}$
\begin{tabbing}
  \hspace*{1mm} \=  \hspace*{12mm}      \= \hspace*{58mm}                   \= \hspace*{12mm} \= \kill
  \hspace*{1mm} \> ${\tt o_{bb}(X,Y)}$  \> ${\tt \leftarrow m\_o_{bb}(X,Y) \wedge \neg p_{bb}(Y,X) \wedge p_{bb}(X,Y)}$
                \> ${\tt p_{bb}(X,Y)}$  \> ${\tt \leftarrow m\_p_{bb}(X,Y) \wedge e(X,Y)}$ \\

                \> ${\tt p_{bb}(X,Y)}$  \> ${\tt \leftarrow m\_p_{bb}(X,Y) \wedge e(X,Z) \wedge p_{bb}(Z,Y)}$
                \> ${\tt m\_p_{bb}(Y,X) \leftarrow  m\_o_{bb}(X,Y)}$\\

                \> ${\tt m\_p_{bb}(X,Y) \leftarrow m\_o_{bb}(X,Y) \wedge \neg p_{bb}(Y,X)}$ \>
                \> ${\tt m\_o_{bb}(X,Y) \leftarrow  m\_s\_o_{bb}(X,Y)}$\\
                \> ${\tt m\_p_{bb}(Z,Y) \leftarrow m\_p_{bb}(X,Y) \wedge e(X,Z)}$
\end{tabbing}
as well as the magic seed fact ${\tt m\_s\_o_{bb}(1,2)}$. The
Magic Sets transformation is sound for stratifiable databases.
However, the resulting rule set may be no more stratifiable (as is
the case in the above example) and more general approaches than
iterated fixpoint computation are needed. For determining the
well-founded model of general logic programs, the alternating
fixpoint computation by Van Gelder~\cite{vg93} or the conditional
fixpoint by Bry~\cite{Bry89} could be used. The application of
these methods, however, is not really efficient because the
specific reason for the unstratifiability of the transformed rule
sets is not taken into account. As an efficient alternative, the
soft stratification concept together with the soft consequence
operator~\cite{beh03} could be used for determining the positive
part of the well-founded model (cf. Section~\ref{Consequence
Operators}).

\subsection[Update Propagation]{Update Propagation}
\label{Update Propagation}
Determining the consequences of base relation changes is essential
for maintaining materialized views as well as for efficiently
checking integrity. Update propagation (UP) methods have been
proposed aiming at the efficient computation of implicit changes
of derived relations resulting from explicitly performed updates
of extensional facts~\cite{kuc91,man94,oli91,prz88}. We present a
specific method for update propagation which fits well with the
semantics of deductive databases and is based on the soft
consequence operator again. We will use the notion \emph{update}
to denote the 'true' changes caused by a transaction only; that
is, we solely consider sets of updates where compensation effects
(i.e., given by an insertion and deletion of the same fact or the
insertion of facts which already existed, for example) have
already been taken into account.

The task of update propagation is to systematically compute the
set of all induced modifications starting from the physical
changes of base data. Technically, this is a set of delta facts
for any affected relation which may be stored in corresponding
delta relations. For each predicate symbol $p \in {\tt pred}(\D)$,
we will use a pair of delta relations $\langle \Delta^{+}_{\tt
p},\Delta^{-}_{\tt p} \rangle$ representing the insertions and
deletions induced on $p$ by an update on $\D$. The initial set of
delta facts directly results from the given update and represents
the so-called UP seeds. They form the starting point from which
induced updates, represented by derived delta relations, are
computed. In our transformation-based approach, so-called
\emph{propagation rules} are employed for computing delta
relations. A propagation rule refers to at least one delta
relation in its body in order to provide a focus on the underlying
changes when computing induced updates. For showing the
effectiveness of an induced update, however, references to the
state of a relation before and after the base update has been
performed are necessary. As an example of this propagation
approach, consider again the rules for relation ${\tt p}$ from
Subsection~\ref{Query Optimization}. The UP rules $\R^{\Delta}$
with respect to insertions into ${\tt e}$ are as follows :

\begin{tabbing}
 \hspace*{6.8cm} \= \kill
 \hspace*{5mm} $\Delta^+_{\tt p}{\tt (X,Y) \leftarrow} \Delta^+_{\tt e}{\tt (X,Y) \wedge} \neg {\tt p^{old}(X,Y)}$\\
 \hspace*{5mm} $\Delta^+_{\tt p}{\tt (X,Y) \leftarrow} \Delta^+_{\tt e}{\tt (X,Z) \wedge p^{new}(Z,Y) \wedge} \neg {\tt p^{old}(X,Y)}$\\
 \hspace*{5mm} $\Delta^+_{\tt p}{\tt (X,Y) \leftarrow} \Delta^+_{\tt p}{\tt (Z,Y) \wedge e^{new}(X,Z) \wedge} \neg {\tt p^{old}(X,Y)}$
\end{tabbing}
For each relation $p$ we use $p^{old}$ to refer to its old state
before the changes given in the delta relations have been applied
whereas $p^{new}$ is used to refer to the new state of $p$. These
state relations are never completely computed but are queried with
bindings from the delta relation in the propagation rule body and
thus act as a test of effectiveness. In the following, we assume
the old database state to be present such that the adornment
\emph{old} can be omitted. For simulating the new database state
from a given update so called $transition\ rules$~\cite{oli91} are
used. The transition rules $\R^{\Delta}_{\tau}$ for simulating the
required new states of ${\tt e}$ and ${\tt p}$ are:
\begin{tabbing}
 \hspace*{6.8cm} \= \kill
 \hspace*{5mm} ${\tt e^{new}(X,Y) \leftarrow e(X,Y) \wedge} \neg \Delta^-_{\tt e}{\tt (X,Y)}$  \>
               ${\tt p^{new}(X,Y) \leftarrow e^{new}(X,Y)}$\\
 \hspace*{5mm} ${\tt e^{new}(X,Y) \leftarrow} \Delta^+_{\tt e}{\tt (X,Y)}$                     \>
               ${\tt p^{new}(X,Y) \leftarrow e^{new}(X,Z) \wedge p^{new}(Z,Y)}$
\end{tabbing}
\noindent Note that the new state definition of intensional
predicates only indirectly refers to the given update in contrast
to extensional predicates. If $\R$ is stratifiable, the rule set
$\R \muplus \R^{\Delta} \muplus \R^{\Delta}_{\tau}$ will be
stratifiable, too (cf.~\cite{BM04}). As $\R \muplus \R^{\Delta}
\muplus \R^{\Delta}_{\tau}$ remains to be stratifiable, iterated
fixpoint computation could be employed for determining the
semantics of these rules and the induced updates defined by them.
However, all state relations are completely determined which leads
to a very inefficient propagation process. The reason is that the
supposed evaluation over the two consecutive database states is
performed using deductive rules which are not specialized with
respect to the particular updates that are propagated. This
weakness of propagation rules in view of a bottom-up
materialization will be cured by incorporating Magic Sets.

\paragraph[Magic Updates]{Magic Updates} \label{Magic Updates}
\hspace*{1mm}\\\\
The aim is to develop an UP approach which is automatically
limited to the affected delta relations. The evaluation of side
literals and effectiveness tests is restricted to the updates
currently propagated. We use the Magic Sets approach for
incorporating a top-down evaluation strategy by considering the
currently propagated updates in the dynamic body literals as
abstract queries on the remainder of the respective propagation
rule bodies. Evaluating these propagation queries has the
advantage that the respective state relations will only be
partially materialized. As an example, let us consider the
specific deductive database $\D=\langle \F,\R, \I \rangle$ with
$\R$ consisting of the well-known rules for the transitive closure
${\tt p}$ of relation ${\tt e}$:
 \begin{tabbing}
 \hspace*{2mm}\underline{$\R$:} \hspace*{2mm} \= ${\tt p(X,Y) \leftarrow e(X,Y)}$\\
                                              \> ${\tt p(X,Y) \leftarrow e(X,Z),p(Z,Y)}$ \\
 \hspace*{10mm}\\
 \hspace*{2mm}\underline{$\F$:} \> {\tt edge(1,2), edge(1,4), edge(3,4) }\\
                                  \> {\tt edge(10,11), edge(11,12), \dots, edge(98,99), edge(99,100)}
 \end{tabbing}
Note that the derived relation ${\tt p}$ consists of $4098$
tuples. Suppose a given update contains the new tuple $e(2,3)$ to
be inserted into $\D$ and we are interested in finding the
resulting consequences for ${\tt p}$. Computing the induced update
by evaluating the stratifiable propagation and transition rules
would lead to the generation of $94$ new state facts for relation
${\tt e}$, $4098$ old state facts for ${\tt p}$ and $4098+3$ new
state facts for ${\tt p}$. The entire number of generated facts is
$8296$ for computing the three induced insertions $\Delta^{+}_{\tt
p}{\tt (1,3)},\Delta^{+}_{\tt p}{\tt (2,3)},\Delta^{+}_{\tt p}{\tt
(2,4)}\}$ with respect to ${\tt p}$.

However, the application of the Magic Updates rewriting with
respect to the propagation queries $\{\Delta^{+}_{\tt p}{\tt
(Z,Y)},\Delta^{+}_{\tt e}{\tt (X,Y)},\Delta^{+}_{\tt e}{\tt
(X,Z)}\}$ provides a much better focus on the changes to ${\tt
e}$. Within its application, the following subquery rules

 \begin{tabbing}
  \hspace*{6.5cm} \= \kill
   \hspace*{5mm} ${\tt m\_p^{new}_{bf}(Z) \leftarrow} \Delta^{+}_{\tt e}{\tt (X,Z)}$ \>
                 ${\tt m\_p_{bb}(X,Y) \leftarrow} \Delta^{+}_{\tt e}{\tt (X,Y)}$ \\
   \hspace*{5mm} ${\tt m\_e^{new}_{fb}(Z) \leftarrow} \Delta^{+}_{\tt p}{\tt (Z,Y)}$ \>
                 ${\tt m\_p_{bb}(X,Y) \leftarrow} \Delta^{+}_{\tt e}{\tt (X,Z) \wedge p^{new}_{bf}(Z,Y)} $ \\
              \> ${\tt m\_p_{bb}(X,Y) \leftarrow} \Delta^{+}_{\tt p}{\tt (Z,Y) \wedge e^{new}_{fb}(X,Z)} $
 \end{tabbing}

\noindent are generated. The respective queries $Q =
\{m\_e^{new}_{fb}, m\_p^{new}_{bf}, \ldots\}$ allow to specialize
the employed transition rules, e.g.

 \begin{tabbing}
 \hspace*{7.5cm} \= \kill
 \hspace*{5mm} ${\tt e^{new}_{fb}(X,Y) \leftarrow m\_e^{new}_{fb}(Y) \wedge e(X,Y) \wedge} \neg \Delta^-_{\tt e}{\tt (X,Y)}$ \\
 \hspace*{5mm} ${\tt e^{new}_{fb}(X,Y) \leftarrow m\_e^{new}_{fb}(Y) \wedge } \Delta^+_{\tt e}{\tt (X,Y)}$
 \end{tabbing}
such that only relevant state tuples are generated. We denote the
Magic Updates transformed rules $\R \muplus \R^{\Delta} \muplus
\R^{\Delta}_{\tau}$ by $\R^{\Delta}_{mu}$. Despite of the large
number of rules in $\R^{\Delta}_{mu}$, the number of derived
results remains relatively small. Quite similar to the Magic sets
approach, the Magic Updates rewriting may result in an
unstratifiable rule set. This is also the case for our example
where the following negative cycle occurs in the respective
dependency graph:
\begin{center}
 $ \Delta^+_{\tt p} \stackrel{pos}{\longrightarrow} {\tt m\_p_{bb}} \stackrel{pos}{\longrightarrow} {\tt p_{bb}} \stackrel{neg}{\longrightarrow} \Delta^+_{\tt p}$
\end{center}
In~\cite{BM04} it has been shown, however, that the resulting
rules must be at least softly stratifiable such that the soft
consequence operator could be used for efficiently computing their
well-founded model. Computing the induced update by evaluating the
Magic Updates transformed rules leads to the generation of two new
state facts for e, one old state fact and one new state fact for
p. The entire number of generated facts is $19$ in contrast to
$8296$ for computing the three induced insertions with respect to
p.

\subsection[View Updates]{View Updates} \label{View Updates}
Bearing in mind the numerous benefits of the afore mentioned
methods to query optimization and update propagation, it seemed
worthwhile to develop a similar, i.e., incremental and
trans\-for\-mation-based, approach to the dual problem of view
updating. In contrast to update propagation, view updating aims at
determining one or more base relation updates such that all given
update requests with respect to derived relations are satisfied
after the base updates have been successfully applied. In the
following, we recall a transformation-based approach to
incrementally compute such base updates for stratifiable databases
proposed by the author in~\cite{BM08}. The approach extends and
integrates standard techniques for efficient query answering,
integrity checking and update propagation. The analysis of view
updating requests usually leads to alternative view update
realizations which are represented in disjunctive form.

\paragraph{Magic View Updates} \label{Magic View Updates}
\hspace*{1mm}\\\\
In our transformation-based approach, true view updates (VU) are
considered only, i.e., ground atoms which are presently not
derivable for atoms to be inserted, or are
derivable for atoms to be deleted, respectively. A method for view
updating determines sets of alternative updates (called VU
realization) satisfying a given request. There may be infinitely
many realizations and even realizations of infinite size which
satisfy a given VU request. In our approach, a breadth-first
search is em\-ployed for determining a set of minimal
realizations. A realization is minimal in the sense that none of
its updates can be removed without losing the property of being a
realization. As each level of the search tree is completely
explored, the result usually consists of more than one
realization. If only VU realizations of infinite size exist, our
method will not terminate.

Given a VU request, view updating methods usually determine
subsequent VU requests in order to find relevant base updates.
Similar to delta relations for UP we will use the notion \emph{VU
relation} to access individual view updates with respect to the
relations of our system. For each relation $p \in {\tt pred}(\R
\cup \F)$ we use the VU relation $\nabla^{+}_{\tt p}(\vec{x})$ for
tuples to be inserted into $\D$ and $\nabla^{-}_{\tt p}(\vec{x})$
for tuples to be deleted from $\D$. The initial set of delta facts
resulting from a given VU request is again represented by
so-called VU seeds. Starting from the seeds, so-called VU rules
are employed for finding subsequent VU requests systematically.
These rules perform a top-down analysis in a similar way as the
bottom-up analysis implemented by the UP rules. As an example,
consider the following database $\D=\langle \F,\R,\I \rangle$ with
$\F=\{r_2(2),s(2)\}$, $\I=\{ic(2)\}$ and the rules $\R$:
\begin{tabbing}
 \hspace*{60mm} \= Musterline \kill
 \hspace*{5mm} ${\tt p(X) \leftarrow q_1(X)}$              \> ${\tt q_1(X) \leftarrow r_1(X) \wedge s(X)}$\\
 \hspace*{5mm} ${\tt p(X) \leftarrow q_2(X)}$              \> ${\tt q_2(X) \leftarrow r_2(X) \wedge} \neg {\tt s(X)}$\\
 \hspace*{5mm} ${\tt ic(2) \leftarrow} \neg {\tt  au(2)}$ \> ${\tt au(X) \leftarrow q_2(X) \wedge} \neg {\tt q_1(X)}$
\end{tabbing}
The corresponding set of VU rules $\R^{\nabla}$ with respect to
$\nabla^{+}_{\tt p}(2)$ is given by:
\begin{tabbing}
 \hspace*{60mm} \= Musterline \kill
 \hspace*{5mm} $\nabla^+_{\tt q_1}{\tt (X) \vee} \nabla^+_{\tt q_1}{\tt (X)} {\tt \leftarrow} \nabla^+_{\tt p}{\tt (X)}$\\
 \hspace*{5mm} $\nabla^+_{\tt r_1}{\tt (X) \leftarrow} \nabla^+_{\tt q_1}{\tt (X) \wedge} \neg {\tt r_1(X)}$
     \> $\nabla^+_{\tt r_2}{\tt (X) \leftarrow} \nabla^+_{\tt q_2}{\tt (X) \wedge} \neg {\tt r_2(X)}$\\
 \hspace*{5mm} $\nabla^+_{\tt s}{\tt (X) \leftarrow} \nabla^+_{\tt q_1}{\tt (X) \wedge} \neg {\tt s(X)}$
     \> $\nabla^-_{\tt s}{\tt (X) \leftarrow} \nabla^+_{\tt q_2}{\tt (X) \wedge s(X)}$
\end{tabbing}
In contrast to the UP rules from Section~\ref{Update Propagation},
no explicit references to the new database state are included in
the above VU rules. The reason is that these rules are applied
iteratively over several intermediate database states before the
minimal set of realizations has been found. Hence, the apparent
re\-fe\-rences to the old state really refer to the current state
which is continuously modified while computing VU realizations.
These predicates solely act as tests again queried with respect to
bindings from VU relations and thus will never be completely
evaluated.

Evaluating these rules using model generation with disjunctive
facts leads to two alternative updates, insertion $\{r_1(2)\}$ and
deletion $\{s(2)\}$, induced by the derived disjunction
$\nabla^+_{\tt r_1}{\tt (2)} \vee \nabla^-_{\tt s}{\tt (2)}$.
Obviously, the second update represented by $\nabla^-_{\tt s}{\tt
(2)}$ would lead to an undesired side effect by means of an
integrity violation. In order to provide a complete method,
however, such erroneous/incomplete paths must be also explored and
side effects repaired if possible. Determining whether a computed
update will lead to a consistent database state or not can be done
by applying a bottom-up UP process at the end of the top-down
phase leading to an irreparable constraint violation with respect
to $\nabla^-_s(2)$:
\begin{tabbing}
 \hspace*{10mm}$\nabla^-_{\tt s}{\tt (2)} \Rightarrow \Delta^+_{\tt q_2}{\tt (2)} \Rightarrow \Delta^+_{\tt p}{\tt (2)},
                \Delta^+_{\tt au}{\tt (2)} \Rightarrow \Delta^-_{\tt ic}{\tt (2)} \rightsquigarrow false$
\end{tabbing}
In order to see whether the violated constraint can be repaired,
the subsequent view update request $\nabla^+_{\tt ic}{\tt (2)}$
with respect to $\D$ ought to be answered. The application of
$\R^{\nabla}$ yields
\begin{tabbing}
 \hspace*{38mm} $ \Rightarrow \nabla^-_{\tt q_2}{\tt (2)}, \nabla^+_{\tt q_2}(2) \rightsquigarrow false$\\
 \hspace*{10mm}$\nabla^+_{\tt ic}{\tt (2)} \Rightarrow \nabla^-_{\tt aux}{\tt (2)} \Updownarrow$\\
 \hspace*{38mm} $ \Rightarrow \nabla^+_{\tt q_1}{\tt (2)} \Rightarrow \nabla^+_{\tt s}{\tt (2)}, \nabla^-_{\tt s}{\tt (2)} \rightsquigarrow false$
\end{tabbing}
showing that this request cannot be satisfied as inconsistent
subsequent view update requests are generated on this path. Such
erroneous derivation paths will be indicated by the keyword
$false$. The reduced set of updates - each of them leading to a
consistent database state only - represents the set of
realizations $\Delta^+_{\tt r_1}{\tt (2)}$.

An induced deletion of an integrity constraint predicate can be
seen as a side effect of an 'erroneous' VU. Similar side effects,
however, can be also found when induced changes to the database
caused by a VU request may include derived facts which had been
actually used for deriving this view update. This effect is shown
in the following example for a deductive database $\D=\langle
\R,\F,\I \rangle$ with $\R=\{{\tt h(X)}{\tt \leftarrow p(X) \wedge
q(X) \wedge i}, {\tt i}{\tt \leftarrow p(X) \wedge} \neg {\tt
q(X)}\}$, $\F=\{{\tt p(1)}\}$, and $\I=\O$. Given the VU request
$\nabla^+_{\tt h}{\tt (1)} $, the overall evaluation scheme for
determining the only realization $\{ \Delta^+_{\tt q}(1),
\Delta^+_{\tt p}(c^{new_1}) \}$ would be as follows:
\begin{tabbing}
 \hspace*{68mm} $ \Rightarrow \nabla^+_{\tt p}{\tt (c^{new_1})}$\\
 \hspace*{5mm}$\nabla^+_{\tt h}{\tt (1)} \Rightarrow \nabla^+_{\tt q}{\tt (1)} \Rightarrow \Delta^+_{\tt q}{\tt (1)} \Rightarrow \Delta^-_{\tt i} \Rightarrow \nabla^+_{\tt i}                  \Updownarrow$\\
 \hspace*{68mm} $ \Rightarrow \nabla^-_{\tt q}{\tt (1)}, \nabla^+_{\tt q}{\tt (1)} \rightsquigarrow false$
\end{tabbing}
The example shows the necessity of compensating side effects,
i.e., the compensation of the 'deletion' $\Delta^-_i$ (that
prevents the 'insertion' $\Delta^+_h(1)$) caused by the tuple
$\nabla^+_q(1)$. In general the compensation of side effects,
however, may in turn cause additional side effects which have to
be 'repaired'. Thus, the view updating method must alternate
between top-down and bottom-up phases until all possibilities for
compensating side effects (including integrity constraint
violations) have been considered, or a solution has been found. To
this end, so-called \emph{VU transition rules}
$\R^{\nabla}_{\tau}$ are used for restarting the VU analysis. For
example, the compensation of violated integrity constraints can be
realized by using the following kind of transition rule
$\Delta^-_{ic}(\vec{c}) \rightarrow \nabla^{+}_{ic}(\vec{c})$ for
each ground literal  $ic(\vec{c}) \in \I$. VU transition rules
make sure that erroneous solutions are evaluated to $false$ and
side effects are repaired.

Having the rules for the direct and indirect consequences of a
given VU request, a general application scheme for systematically
determining VU realizations can be defined (see\cite{BM08} for
details). Instead of using simple propagation rules $\R \muplus
\R^{\Delta} \muplus \R^{\Delta}_{\tau}$, however, it is much more
efficient to employ the corresponding Magic Update rules. The
top-down analysis rules $\R \muplus \R^{\nabla}$ and the bottom-up
consequence analysis rules $\R^{\Delta}_{mu} \muplus
\R^{\nabla}_{\tau}$ are alternating applied. Note that the
disjunctive rules $\R \muplus \R^{\nabla}$ are stratifiable while
$\R^{\Delta}_{mu} \muplus \R^{\nabla}_{\tau}$ is softly
stratifiable such that a perfect model state~\cite{beh07,FM92} and
a well-founded model generation must alternately be applied. The
iteration stops as soon as a realization for the given VU request
has been found. The correctness of this approach has been already
shown in~\cite{BM08}.

\section[Consequence Operators]{Consequence Operators and Fixpoint Computations}
\label{Consequence Operators}
In the following, we summarize the most important fixpoint-based
approaches for definite as well as indefinite rules. All these
methods employ so-called consequence operators which formalize the
application of deductive rules for deriving new data. Based on
their properties, a new uniform consequence operator is developed
subsequently.

\subsection[Definite Rules]{Definite Rules}
\label{Definite Rules}
First, we recall the iterated fixpoint method for constructing the
well-founded model of a stratifiable database which coincides with
its perfect model~\cite{prz88}.

\begin{definition} Let $\D=\langle \F,\R\rangle$ be a deductive database,
$\lambda$ a stratification on $\D$, $\R_1 \muplus \ldots \muplus
\R_n$ the partition of $\R$ induced by $\lambda$, $I \subseteq
\He_{\D}$ a set of ground atoms, and $\mathbb{[[} \R
\mathbb{]]}_{I}$ the set of all ground instances of rules in $\R$
with respect to the set $I$. Then we define
\begin{enumerate}
 \item the immediate consequence operator $T_{\R}(I)$ as
       \begin{tabbing}
       \hspace*{28mm} \= \kill
       \hspace*{2mm} $T_\R(I) := \{ H\ |$
        \> $H \in I \vee \exists r \in \mathbb{[[} \R \mathbb{]]}_{I}:$ $r \equiv H \leftarrow L_1 \wedge \ldots \wedge L_n$\\
        \> such that $L_i \in I $ for all positive literals $L_i$\\
        \> and $L \notin I$ for all negative literals $L_j \equiv \neg L\}$,
       \end{tabbing}

 \item the iterated fixpoint $M_n$ as the last Herbrand model of
       the sequence
       \begin{tabbing}
       \hspace*{2mm} $M_1$ := {\tt lfp} $(T_{\R_1}, \F )$, $M_2$ := {\tt lfp} $(T_{\R_2}, M_1 )$, \ldots, $M_n$ := {\tt lfp} $(T_{\R_n}, M_{n-1} )$,
       \end{tabbing}
       where {\tt lfp} $(T_\R,\F)$ denotes the least fixpoint of operator $T_\R$
       containing $\F$.

 \newpage

 \item and the iterated fixpoint model $\Mo^i_{\D}$ as
       \begin{tabbing}
       \hspace*{2mm} $\Mo^i_{\D} :=$ $M_n \muplus \neg \cdot \overline{M_n}$.
       \end{tabbing}
\end{enumerate}
\end{definition}
This constructive definition of the iterated fixpoint model is
based on the immediate consequence operator introduced by van
Emden and Kowalski. In~\cite{prz88} it has been shown that the
perfect model of a stratifiable database $\D$ is identical with
the iterated fixpoint model $\Mo^i_{\D}$ of $\D$.

Stratifiable rules represent the most important class of deductive
rules as they cover the expressiveness of recursion in SQL:1999.
Our transformation-based approaches, however, may internally lead
to unstratifiable rules for which a more general inference method
is necessary. In case that unstratifiability is caused by the
application of Magic Sets, the so-called soft stratification
approach proposed by the author in~\cite{beh03} could be used.

\begin{definition} Let $\D=\langle \F,\R\rangle$ be a deductive database,
$\lambda^s$ a soft stratification on $\D$, $\Pa= P_1 \muplus
\ldots \muplus P_n$ the partition of $\R$ induced by $\lambda^s$,
and $I \subseteq \He_{\D}$ a set of ground atoms. Then we define
\begin{enumerate}
 \item the soft consequence operator $T^s_{\Pa}(I)$ as
       \begin{tabbing}
       \hspace*{28mm} \= \kill
       \hspace*{2mm} $T^s_{\Pa}(I) :=\left \{
                \begin{array}{l@{\quad}l}
                  I                & if\ T_{P_j}(I) = I\ forall\ j\in\{1, \ldots,n\}\\
                  T_{P_i}(I)       & with\ i\ =\ {\tt min} \{ j\ |\ T_{P_j}(I) \supsetneq I \},\ otherwise.
                \end{array} \right.\)
       \end{tabbing}
       where $T_{P_i}$ denotes the immediate consequence operator.

 \item and the soft fixpoint model $\Mo^s_{\D}$ as
       \begin{tabbing}
       \hspace*{2mm} $\Mo^s_{\D} :=$ ${\tt lfp}\ (T^s_\Pa,\F) \muplus \neg \cdot \overline{{\tt ( lfp}\ (T^s_\Pa,\F){\rm )}}$.
       \end{tabbing}
\end{enumerate}
\end{definition}
Note that the soft consequence operator is based upon the
immediate consequence operator and can even be used to determine
the iterated fixpoint model of a stratifiable
database~\cite{BM04}. As an even more general alternative, the
alternating fixpoint model for arbitrary unstratifiable rules has
been proposed in~\cite{KSS95} on the basis of the eventual
consequence operator.

\begin{definition} Let $\D=\langle \F,\R \rangle$ be a deductive database,
$I^+,I^- \subseteq \He_{\D}$ sets of ground atoms, and
$\mathbb{[[} \R \mathbb{]]}_{I^+}$ the set of all ground instances
of rules in $\R$ with respect to the set $I^+$. Then we define
\begin{enumerate}
 \item the eventual consequence operator $\widehat{T}_{\R}\langle I^-\rangle$ as
       \begin{tabbing}
       \hspace*{35mm} \= \kill
       \hspace*{2mm} $\widehat{T}_\R \langle I^-\rangle (I^+) := \{ H\ |$
        \> $H \in I^+ \vee \exists r \in \mathbb{[[} \R \mathbb{]]}_{I^+}:$ $r \equiv H \leftarrow L_1 \wedge \ldots \wedge L_n$\\
        \> such that $L_i \in I^+ $ for all positive literals $L_i$\\
        \> and $L \notin I^-$ for all negative literals $L_j \equiv \neg L\}$,
       \end{tabbing}

 \item the eventual consequence transformation $\widehat{S}_{\D}$ as
       \begin{tabbing}
       \hspace*{2mm} $\widehat{S}_{\D}(I^-) := \ {\tt lfp} (\widehat{T}_{\R} \langle I^-\rangle, \F)$,
       \end{tabbing}

 \item and the alternating fixpoint model $\Mo^a_{\D}$ as
       \begin{tabbing}
       \hspace*{2mm} $\Mo^a_{\D} :=$ ${\tt lfp}\ (\widehat{S}^2_{\D},\O) \muplus \neg \cdot \overline{\widehat{S}^2_{\D}{\tt ( lfp}\ (\widehat{S}^2_{\D},\O){\rm )}}$ ,
       \end{tabbing}
       where $\widehat{S}^2_{\D}$ denotes the nested application
       of the eventual consequence transformation, i.e.,
       $\widehat{S}^2_{\D}(I^-) = \widehat{S}_{\D}(\widehat{S}_{\D}(I^-))$.
\end{enumerate}
\end{definition}
In~\cite{KSS95} it has been shown that the alternating fixpoint
model $\Mo^a_\D$ coincides with the well-founded model of a given
database $\D$. The induced fixpoint computation may indeed serve
as a universal model generator for arbitrary classes of deductive
rules. However, the eventual consequence operator is
computationally expensive due to the intermediate determination of
supersets of sets of true atoms. With respect to the discussed
transformation-based approaches, the iterated fixpoint model could
be used for determining the semantics of the stratifiable subset
of rules in $\R_{ms}$ for query optimization, $\R^{\Delta}_{mu}$
for update propagation, and $\R^{\Delta}_{mu} \muplus
\R^{\nabla}_{\tau}$ for view updating. If these rule sets contain
unstratifiable rules, the soft or alternating fixpoint model
generator ought be used while the first has proven to be more
efficient than the latter~\cite{beh03}. None of the above
mentioned consequence operators, however, can deal with indefinite
rules necessary for evaluating the view updating rules $\R \muplus
\R^{\nabla}$.

\subsection{Indefinite Rules}
\label{Indefinite Rules}
In~\cite{beh07}, the author proposed a consequence operator for
the efficient bottom-up state generation of stratifiable
disjunctive deductive databases. To this end, a new version of the
immediate consequence operator based on hyperresolution has been
introduced which extends Minker's operator for positive
disjunctive Datalog rules~\cite{Mink82}. In contrast to already
existing model generation methods our approach for efficiently
computing perfect models is based on state generation. Within this
disjunctive consequence operator, the mapping ${\tt red}$ on
indefinite facts is employed which returns non-redundant and
subsumption-free representations of disjunctive facts.
Additionally, the mapping ${\tt min\_models}(F)$ is used for
determining the set of minimal Herbrand models from a given set of
disjunctive facts $F$. We identify a disjunctive fact with a set
of atoms such that the occurrence of a ground atom A within a fact
$f$ can also be written as $A \in f$. The set difference operator
can then be used to remove certain atoms from a disjunction while
the empty set as result is interpreted as $false$.

\begin{definition} Let $\D=\langle \F,\R \rangle$ be a stratifiable disjunctive
database rules,$\lambda$ a stratification on $\D$, $\R_1 \muplus
\ldots \muplus \R_n$ the partition of $\R$ induced by $\lambda$,
$I$ an arbitrary subset of indefinite facts from the disjunctive
Herbrand base~\cite{FM92} of $\D$, and $\mathbb{[[} \R
\mathbb{]]}_{I}$ the set of all ground instances of rules in $\R$
with respect to the set $I$ Then we define.
\begin{enumerate}
 \item the disjunctive consequence operator $T^{state}_\R$ as
       \begin{tabbing}
       \hspace*{25mm} \= \kill
       \hspace*{-1mm} $T^{state}_\R(I):= {\tt red}(\{ H \ |\ H \in I \vee \exists r \in \mathbb{[[} \R \mathbb{]]}_{I}:$ $r \equiv A_1 \vee \ldots \vee A_l \leftarrow L_1 \wedge \ldots \wedge L_n$\\
                                      \> with $H = (A_1 \vee \dots \vee A_l \vee f_1 \setminus L_1 \vee \dots \vee f_n \setminus L_n \vee C)$\\
                                      \> such that $f_i \in I \wedge L_i \in f_i$ for all positive literals $L_i$\\
                                      \> and $L_j \notin I$ for all negative literals $L_j \equiv \neg L$ \\
                                      \> and $(L_j \in C \Leftrightarrow \exists \M \in {\tt min\_models}(I):$\\
                                                              \> \hspace*{1.5cm} $L_j \in \M$ for at least one negative literal $L_j$ \\
                                                              \> \hspace*{1.5cm} and $L_{k}\in \M$ for all positive literals $L_k$ \\
                                                              \> \hspace*{1.5cm} and $A_l \notin \M$ for all head literals of r$)\})$
       \end{tabbing}
       %where C denotes
 \item the iterated fixpoint state $S_n$ as the last minimal model state of
       the sequence
       \begin{tabbing}
       \hspace*{-1mm} $S_1$ := {\tt lfp} $(T^{state}_{\R_1}, \F )$, $S_2$ := {\tt lfp} $(T^{state}_{\R_2}, S_1 )$, \ldots, $S_n$ := {\tt lfp} $(T^{state}_{\R_n}, S_{n-1} )$,
       \end{tabbing}

 \item and the iterated fixpoint state model $\M\St_{\D}$ as
       \begin{tabbing}
       \hspace*{-1mm} $\M\St_{\D} :=$ $S_n \muplus \neg \cdot \overline{S_n}$.
       \end{tabbing}
 \end{enumerate}
\end{definition}
In~\cite{beh07} it has been shown that the iterated fixpoint state
model $\M\St_{\D}$ of a disjunctive database $\D$ coincides with
the perfect model state of $\D$. It induces a constructive method
for determining the semantics of stratifiable disjunctive
databases. The only remaining question is how integrity
constraints are handled in the context of disjunctive databases.
We consider again definite facts as integrity constraints, only,
which must be derivable in every model of the disjunctive
database. Thus, only those models from the iterated fixpoint state
are selected in which the respective definite facts are derivable.
To this end, the already introduced keyword \emph{false} can be
used for indicating and removing inconsistent model states. The
database is called consistent iff at least one consistent model
state exists.

This proposed inference method is well-suited for determining the
semantics of stratifiable disjunctive databases with integrity
constraints. And thus, it seems to be suited as the basic
inference mechanism for evaluating view updating rules. The
problem is, however, that the respective rules contain
unstratifiable definite rules which cannot be evaluated using the
inference method proposed above. Hence, the evaluation techniques
for definite (Section~\ref{Definite Rules}) and indefinite rules
(Section~\ref{Indefinite Rules}) do not really fit together and a
new uniform approach is needed.

\section[A Uniform Fixpoint Approach]{A Uniform Fixpoint Approach}
\label{A Uniform Fixpoint Approach}
In this section, a new version of the soft consequence operator is
proposed which is suited as efficient state generator for softly
stratifiable definite as well as stratifiable indefinite
databases. The original version of the soft consequence operator
$T^{s}_{\Pa}$ is based on the immediate consequence operator by
van Emden and Kowalski and can be applied to an arbitrary
partition $\Pa$ of a given set of definite rules. Consequently,
its application does not always lead to correct derivations. In
fact, this operator has been designed for the application to
softly stratified rules resulting from the application of Magic
Sets. However, this operator is also suited for determining the
perfect model of a stratifiable database.

\begin{lemma}\label{Lemma1}
Let $\D = \langle \F,\R \rangle$ be a stratifiable database and
$\lambda$ a stratification of $\R$ inducing the partition $\Pa$ of
$\R$. The perfect model $\Mo_{\D}$ of $\langle \F,\R \rangle$ is
identical with the soft fixpoint model of $\D$, i.e.,
\begin{tabbing}
  \hspace*{5mm} $\Mo_{\D} = {\tt lfp}(T^s_{\Pa},\F) \muplus \neg \cdot \overline{{\tt lfp}(T^s_{\Pa},\F)}$.
\end{tabbing}
\end{lemma}

\begin{proof}
This property follows from the fact that for every partition $\Pa=
P_1 \muplus \ldots P_n$ induced by a stratification, the condition
${\tt pred}(P_i) \cap {\tt pred}(P_j) = \O$ with $i \neq j$ must
necessarily hold. As soon as the application of the immediate
consequence operator $T_{P_i}$ with respect to a certain layer
$P_i$ generates no new facts anymore, the rules in $P_i$ can never
fire again. The application of the incorporated ${\tt min}$
function then induces the same sequence of Herbrand models as in
the case of the iterated fixpoint computation.\hfill $\Box$
\end{proof}
Another property we need for extending the original soft
consequence operator is about the application of $T^{state}$ to
definite rules and facts.

\begin{lemma}\label{Lemma2}
Let $r$ be an arbitrary definite rule and $f$ be a set of
arbitrary definite facts. The single application of $r$ to $f$
using the immediate consequence operator or the disjunctive
consequence operator, always yields the same result, i.e.,
\begin{tabbing}
  \hspace*{5mm} $T_{r}(f)\ =\ T^{state}_{r}(f)$.
\end{tabbing}
\end{lemma}
\begin{proof}
The proof follows from the fact that all non-minimal conclusions
of $T^{state}$ are immediately eliminated by the subsumption
operator ${\tt red}$. \hfill $\Box$
\end{proof}
The above proposition establishes the relationship between the
definite and indefinite case showing that the disjunctive
consequence operator represents a generalization of the immediate
one. Thus, its application to definite rules and facts can be used
to realize the same derivation process as the one performed by
using the immediate consequence operator. Based on the two
properties from above, we can now consistently extend the
definition of the soft consequence operator which allows its
application to indefinite rules and facts, too.

\begin{definition}
Let $\D=\langle \F,\R \rangle$ be an arbitrary disjunctive
database, $I$ an arbitrary subset of indefinite facts from the
disjunctive Herbrand base of $\D$, and $\Pa=P_1 \muplus \ldots
\muplus P_n$  a partition of $\R$. The general soft consequence
operator $T^g_{\Pa}(I)$ is defined as
       \begin{tabbing}
       \hspace*{28mm} \= \kill
       \hspace*{2mm} $T^g_{\Pa}(I) :=\left \{
                \begin{array}{l@{\quad}l}
                  I                & if\ T_{P_j}(I) = I\ forall\ j\in\{1, \ldots,n\}\\
                  T^{state}_{P_i}(I)       & with\ i\ =\ {\tt min} \{ j\ |\ T^{state}_{P_j}(I) \supsetneq I \},\ otherwise.
                \end{array} \right.\)
       \end{tabbing}
       where $T^{state}_{P_i}$ denotes the disjunctive consequence operator.

\end{definition}
In contrast to the original definition, the general soft
consequence operator is based on the disjunctive operator
$T^{state}_{P_i}$ instead of the immediate consequence operator.
The least fixpoint of $T^{g}_\Pa$ can be used to determine the
perfect model of definite as well as indefinite stratifiable
databases and the well-founded model of softly stratifiable
definite databases.

\begin{theo}
Let $\D = \langle \F,\R \rangle$ be a stratifiable disjunctive
database and $\lambda$ a stratification of $\R$ inducing the
partition $\Pa$ of $\R$. The perfect model state $\Pe\St_{\D}$ of
$\langle \F,\R \rangle$ is identical with the least fixpoint model
of $T^g_{\Pa}$, i.e.,
\begin{tabbing}
  \hspace*{5mm} $\Pe\St_{\D} = {\tt lfp}(T^g_{\Pa},\F) \muplus \neg \cdot \overline{{\tt lfp}(T^g_{\Pa},\F)}$.
\end{tabbing}
\end{theo}

\begin{proof}
The proof directly follows from the correctness of the fixpoint
computations for each stratum as shown in~\cite{beh07} and the
same structural argument already used in Lemma~\ref{Lemma1}.\hfill
$\Box$
\end{proof}
The definition of ${\tt lfp}(T^g_{\Pa},\F)$ induces a constructive
method for determining the perfect model state as well as the
well-founded model of a given database. Thus, it forms a suitable
basis for the evaluation of the rules $\R_{ms}$ for query
optimization, $\R^{\Delta}_{mu}$ for update propagation, and
$\R^{\Delta}_{mu} \muplus \R^{\nabla}_{\tau}$ as well as $\R
\muplus \R^{\nabla}$ for view updating. This general approach to
defining the semantics of different classes of deductive rules is
surprisingly simple and induces a rather efficient inference
mechanism in contrast to general well-founded model generators.
The soft stratification concept, however, is not yet applicable to
indefinite databases because ordinary Magic Sets can not be used
for indefinite clauses. Nevertheless, the resulting extended
version of the soft consequence operator can be used as a uniform
basis for the evaluation of all transformation-based techniques
mentioned in this paper.

%%%%%%%%%%%%%%%%%%%%%%%%%%%%%%%%%%%%%%%%%%%%%%%%%%%%%%%%%%%%%%%%%%%%%%%%%%%%%
%
%                  Conclusion
%
%%%%%%%%%%%%%%%%%%%%%%%%%%%%%%%%%%%%%%%%%%%%%%%%%%%%%%%%%%%%%%%%%%%%%%%%%%%%%
%
\section[Conclusion]{Conclusion}
\label{Conclusion}
In this paper, we have presented an extended version of the soft
consequence operator for the efficient top-down and bottom-up
reasoning in deductive databases. This operator allows for the
efficient evaluation of softly stratifiable incremental
expressions and stratifiable disjunctive rules. It solely
represents a theoretical approach but provides insights into
design decisions for extending the inference component of
commercial database systems. The relevance and quality of the
transformation-based approaches, however, has been already shown
in various practical research projects (e.g.~\cite{BDMS08,BSM10})
at the University of Bonn.

%%%%%%%%%%%%%%%%%%%%%%%%%%%%%%%%%%%%%%%%%%%%%%%%%%%%%%%%%%%%%%%%%%%%%%%%%%%%%
%
%       Bibliography
%
%%%%%%%%%%%%%%%%%%%%%%%%%%%%%%%%%%%%%%%%%%%%%%%%%%%%%%%%%%%%%%%%%%%%%%%%%%%%%

\bibliographystyle{diss}
%\bibliography{isbib}

\end{document}